\documentclass{llncs}
\usepackage{color}
\usepackage{graphicx}
\usepackage{amsfonts}
\usepackage{latexsym,amssymb,amsmath}
\usepackage{todonotes}
\usepackage[ruled]{algorithm}
\usepackage[noend]{algpseudocode}
\algdef{SE}[DOWHILE]{Do}{doWhile}{\algorithmicdo}[1]{\algorithmicwhile\ #1}
\usepackage{comment}
\usepackage{epstopdf}
\usepackage{epsfig}
\usepackage{soul}
\usepackage{tikz}
\usepackage{cleveref}

\begin{document}

\title{Bike Assisted Evacuation on a Line of Robots with S/R Communication Faults\\}
\author{
Khaled Jawhar\inst{1}
\and
Evangelos Kranakis\inst{1}\inst{2}
}

\institute{
School of Computer Science, Carleton University, Ottawa, Ontario, Canada.
\and
Research supported in part by NSERC Discovery grant.
}
\maketitle
\hspace{-1cm}
\begin{abstract}
Two autonomous mobile robots and a non-autonomous one, also called bike, are placed at the origin of an infinite line. The autonomous robots can travel with maximum speed $1$. When a robot rides the bike its speed increases to $v>1$, however only exactly one robot at a time can ride the bike and the bike is non-autonomous in that it cannot move on its own. An Exit is placed on the line at an unknown location and at distance $d$ from the origin. The robots have limited communication behavior; one robot is a sender (denoted by S) in that it can send information wirelessly at any distance and receive messages only in F2F (Face-to-Face), while the other robot is a receiver (denoted by R) in that it can receive information wirelessly but can send information only F2F. The bike has no communication capabilities of its own. We refer to the resulting communication model of the ensemble of the two autonomous robots and the bike as S/R. 

Our general goal is to understand the impact of the non-autonomous robot in assisting the evacuation of the two autonomous faulty robots. Our main contribution is to provide a new evacuation algorithm that enables both robots to evacuate from the unknown Exit in the S/R model. We also analyze the resulting evacuation time as a function of the bike's speed $v$ and give upper and lower bounds on the competitive ratio of the resulting algorithm for the entire range of possible values of $v$. 

\vspace{0.5cm}

\end{abstract}

\section{Introduction}

Evacuation (also known as group search) is similar to search except that it involves many entities which cooperate in order to find an unknown exit. There are plenty of applications related to search and evacuation in a distributed system such as data mining, crawling, and surveillance which makes it an area of interest. Many researchers have studied the linear search and evacuation problems with one or more robots moving at different speeds. Most investigations in this domain highlight algorithms that achieve the best possible upper and lower bounds. There are several factors that may affect the solution of the linear search and evacuation problems. For example, the robot(s) may know how far the exit is from the origin but may not be aware of the direction (to the left or right of the origin), the robot(s) may know the direction but not the distance from the origin to the exit, or in a more complicated case the distance and the direction are unknown to the robots. In our present study, the distance and the direction are unknown and there is a communication fault affecting both robots.

\subsection{Preliminaries, Notation and Terminology}

The system consists of two autonomous mobile agents (referred to as robots) and non-autonomous one (referred to as bike). The two robots and the bike together are sometimes referred to as the ensemble. The two robots have different identities and can see each other when they are situated at the same position. The search domain is the infinite line and it is bidirectional in that the robots can move in either direction on the line without this affecting their speeds. The two autonomous mobile agents can move around on their own with maximum speed $1$ or ride the bike with speed $v>1$. The robots can change direction at any time as specified by an algorithm; further, this change in direction of a robot  is instantaneous, with or without the bike). An autonomous robot riding the bike is also referred to as {\em biker}, otherwise it also referred to as {\em hiker}. 

An evacuation algorithm is given by the complete description of the trajectories traced by the two bikers until they both find the exit and evacuate.  We are interested in evacuation algorithms which achieve the best competitive ratio which is defined as the optimal ratio of the evacuation time achieved by an algorithm which enables both autonomous robots to evacuate from the exit when its location is unknown divided by the evacuation time needed if all the robots know where the exit is.

The bike is not autonomous and cannot move and/or communicate on its own and thus plays only the role of assistant in order to speed up the search. The autonomous robot using the bike has an advantage in that it can move with speed $v > 1$ which is of course faster than its walking speed $1$. Since the bike is a limited resource, in that it can be used by only one robot at a time, it must be shared by the robots in the sense that they must take turns using it. This is easily seen to be an advantage since it will ultimately improve the overall evacuation time and hence also the competitive ratio of the ensemble.

An important aspect in our algorithms will be ``bike switching'', by which we mean changing the rider of the bike. We will assume throughout the paper that bike switching between robots is instantaneous and at no time cost. Note that the robots may recognize the presence of the bike when they are at the same location as the bike.

We employ a communication model referred to as S/R model (this model was first proposed on the infinite line by~\cite{czyzowicz2021groupevac} and on a unit disk by~\cite{georgiou2022evacuation}) in which the robots have certain limited communication behavior and may communicate throughout the execution of the algorithms in the following way. Face-To-Face (F2F) (resp. Wireless) communication can take place only when the robots are co-located (resp. at any distance). The sender can only send information wirelessly and receive information F2F while the receiver can only receive information wirelessly and send information F2F. A typical communication exchange may involve, e.g., ``exit is found'', ``bike released'', ``switch bike'', etc. Note that the robots are endowed with pedometers and have computing abilities so that they can deduce the location of the other robot and/or the bike from relevant communications exchanged. Moreover, their ability to communicate F2F never fails. The bike itself has no communication capabilities on its own. In a way we think of the S/R model as describing a type of fault and in the sequel, we refer to the ensemble of the two robots and the bike as the S/R (communication) model. 

Throughout the paper we will be using S and R to denote the sender and receiver respectively (which we often refer to as bikers) and $B$ to denote the bike (non-autonomous robots). The robots are equipped with pedometers which can be used either when walking or riding the bike and are identical in all their capabilities (locomotion). The origin of the real line will be at the point $x=0$ on the $x$-axis and this will also be the starting location of the robots and the bike. The adversary may place the exit at either of the points $\pm d$, where $d > 0$ will denote the unknown distance of the exit from the origin. In addition, $v > 1$ will denote the speed that a biker can attain when riding the bike.

The following lemma is shown in \cite{jawharkranakis21} can be proved easily using the standard analysis of bike-sharing and will be used extensively in the sequel.
\begin{lemma}
\label{lm:bikeshare}
Assume two robots of maximum speed $1$ are sharing a bike of speed $v > 1$. Together the ensemble can cover a segment of length $d$ in time $\frac{d (v+1)}{2 v}$. In fact the ensemble travels with speed equal to $\frac {2 v}{v+1}$. Moreover, these time and speed values are optimal.
\end{lemma}

\subsection{Related work}
The first search problems in the literature concerned a single robot and were proposed for a continuous infinite line by Beck~\cite{beck1964linear} and Bellman \cite{bellman1963optimal} with the focus on stochastic search models and their analysis. Assuming that the distance and the direction to the exit are unknown, they proposed an optimal algorithm with competitive ratio $9$. Additional research in that field can be found in ~\cite{ahlswede1987search,stone1975theory}. There have been many variants of the search problem such as having static or moving targets, involving multiple robots with limited communication behavior, or robots with differing speeds, etc. 

Search and evacuations problems were studied in environments with multiple distinct speed robots \cite{bampas2019linear,isaacCzyzowiczGKKNOS16,PODC16} as well as in various domains including disks, triangles, and circles \cite{brandt2017collaboration,czyzowicz2018evacuating,czyzowicz2019groupkos}. Evacuation algorithms for two robots on a line without bike assistance are known for two robots without limited communication behavior and with different speeds $v$ and $1$ such that $v < 1$; for example, there is an evacuation algorithm in the F2F model with optimal competitive ratio of $\frac{1+3 v}{1-v}$, for $v\leq \frac 13$, and $9$ otherwise, see~\cite{bampas2019linear}. The same paper also considers the wireless communication model. Our work differs from~\cite{bampas2019linear} in two aspects: first we are using a mixed communication model S/R, and second a bike (non-autonomous robot) is present which is shared by the two autonomous robots and has the effect of increasing the overall evacuation time of the ensemble. The competitive ratio of evacuation for two robots without communication fault using the F2F model is known to be $9$, see\cite{czyzowicz2019groupkos}. Additional studies on linear search can be found in the works of 
\cite{baezayates1993searching,BS95} and under various models of linear search
concerning search cost and robot communication in \cite{Bose16,chrobak2015group,demaine2006online}.

An interesting variant of the original linear search problem is concerned with the case where some of the robots are faulty (crash or byzantine). There are several research works on this theme. The two main papers in this line of research are \cite{PODC16} for crash-faulty robots and  \cite{isaacCzyzowiczGKKNOS16} for Byzantine-faulty robots. Moreover, \cite{czyzowicz2021searchnew} studies search and evacuation with a near majority of faulty agents (e.g., three robots at most one of which is byzantine faulty)--which also exhibits the worst-case evacuation time. In addition, the S/R mixed communication model (also considered in our current paper) was  considered as a way to model group search in the presence of robots with faulty communication capabilities and was first introduced for an infinite line by~\cite{czyzowicz2021groupevac}; it was shown that there is an evacuation algorithm with competitive ratio of $3 + 2 \sqrt{2}$ and this is optimal; however, this paper does not include the concept of the non-autonomous robot considered here. 

We should also mention that the S/R faulty communication model considered in our present paper was also studied for a unit disk in the recent paper \cite{georgiou2022evacuation}. Additional research on evacuation from a unit disk in the presence of a (byzantine) faulty robot was first studied in \cite{czyzowicz2017evacuation} and further elaborated for crash and byzantine faults in \cite{georgiou2019optimal}. 

Bike-assisted search and evacuation on a line was first considered for two robots without faults in~\cite{jawharkranakis21} and soon thereafter \cite{georgiou2022evacuation} for a unit disk. To the best of our knowledge the present study of Bike-assisted search and evacuation on a line for two robots and a bike in the S/R model has not been considered in the search literature.

\subsection{Outline and results of the paper}

In the sequel, in Section~\ref{sec:Evacuation in the S/R Model} we give the main evacuation algorithm in three parts: Subsubsection~\ref{subsec:alg1} includes the algorithm when $v \in [1,3]$, Subsection~\ref{subsec:alg2} the algorithm when $v \in [3,10]$, and Subsection~\ref{subsec:alg3} the algorithm when $v \in [10, \infty)$, while in Subsection~\ref{subsec:alg123} we plot the evacuation time for the entire range of the bike's speed $v$ so that we can gauge the performance of the main algorithm. Table~\ref{table1} displays in detail the competitive ratios in each of the algorithms for the upper bound cases, respectively, based on the speed $v$ of the biker.  
\begin{table}[htp]
\caption{Upper bound on the competitive ratio for the three different algorithms; the first column gives the algorithm, the second column the speed $v$ of the biker for which the respective upper bound on the competitive ratio in the third column is valid.}
\begin{center}
\begin{tabular}{| l | c | c|}
\hline
Algorithm & Bike Speed & Competitive Ratio \\
\hline
Algorithm~\ref{algorithm1} & $v \in [1,3]$ & $\frac{2v}{v+1}  \left(\frac{2 v+\frac{-(7v+v^2)+v\sqrt{v^2+30 v+97}}{2 v+6}}{v \frac{-(7v+v^2)+v\sqrt{v^2+30 v+97}}{2 v+6}}\right)$\\
\hline
Algorithm~\ref{algorithm2} & $v \in [3,10]$ & $\frac{2v}{v+1} \left(1+\frac{1}{v}+\frac{v^{2}-3 v-2-\sqrt{v^{4}+18 v^{3}-7 v^{2}+4 v+4}}{4 v (1-v)}\right)$\\
\hline
Algorithm~\ref{algorithm3} & $v \in [10,\infty]$ & $\frac{2 v}{v+1} \left(\frac{9}{v}+\frac{1}{2}-\frac{1}{2v^2}\right)$\\
\hline
\end{tabular}
\end{center}
\label{table1}
\end{table}

We prove a lower bound in Section~\ref{sec:Lower Bound}, 
and in Table~\ref{table2} we display the competitive ratios in each of the algorithms for the lower bound cases, respectively, based on the interval to which the speed $v$ of the biker belongs to.  
\begin{table}[htp]
\caption{Lower bound on the competitive ratio for the evacuation of the two robots for the given range on the speed $v$ of the biker.}
\begin{center}
\begin{tabular}{|c| c|}
\hline
 Bike Speed $v$ & Competitive Ratio \\
\hline
$v \leq 3$ & $\frac{v^{2}+2 v-3}{v^{2}-1}$ \\
\hline
$v > 3$ & $\frac{6}{v+1}$ \\
\hline
\end{tabular}
\end{center}
\label{table2}
\end{table}
Finally we conclude in Section~\ref{sec:Conclusion} with a summary and a discussion of additional ideas for research. 

Due to space limitations all missing proofs can be found in the appendix.

\section{Evacuation in the S/R Model}
\label{sec:Evacuation in the S/R Model}

In this section we describe the algorithms in order to obtain upper bounds for two robots in the S/R model. We also analyze the resulting competitive ratios. Recall that the competitive ratio of an algorithm is defined as the evacuation time of the algorithm divided by the shortest time required by both robots to reach the exit if they know where the exit is. 

 In order to reach the exit in shortest time possible, there are two cases to consider. In the first case, the two robots choose to move in opposite direction and the one which finds the exit first will notify the other one (see Subsection~\ref{subsec:oppdir}). In the second case, the sender takes the bike and uses a Zig-Zag strategy and the receiver imitates the sender using its maximum unit speed (see Subsection~\ref{subsec:alg3}).  Let us study the algorithms proposed in each of the two cases.

\subsection{Robots move in opposite direction}
\label{subsec:oppdir}

Let us consider the case that the sender $S$ and the receiver $R$ are moving in opposite direction. To this end we will ``reduce'' the speeds of the sender and the receiver and our analysis will show the choice that will yield the shortest evacuation time. Let us assume that $S$ moves with speed $0 \leq u_{1} \leq 1$ and $R$ takes the bike and moves with speed $0\leq u_{2}\leq v$. As already proposed, the objective is to find the optimal speeds $u_{1}$ and $u_{2}$ that should be used by $S$ and $R$, respectively, in order for the resulting evacuation time to be minimal. There are two cases to take into account depending on which of the two autonomous robots finds the exit first:

\paragraph{Case~1:} $S$ finds the exit first.

As soon as S reaches the exit, it communicates with R which is on the other side to proceed to the exit. Based on Fig. \ref{fig2}, 
\begin{figure}[h] 
\begin{center}
\epsfig{width=11cm,file=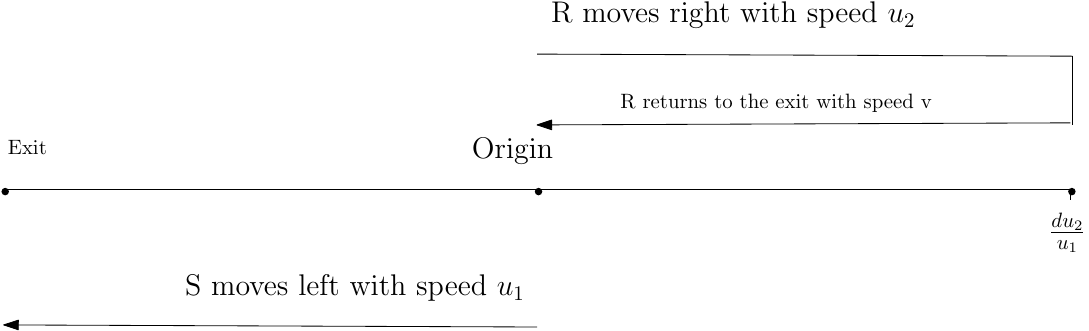}
\end{center}
     \caption{Depicted is Case 1 where two robots S and R are moving in opposite direction such that S finds the exit first.}
       \label{fig2}
\end{figure} 
S needs time $\frac{d}{u_{1}}$ to reach the exit. When S reaches the exit, R would be at distance $\frac{d u_{2}}{u_{1}}$ on the other side and it needs time $\frac{d u_{2}}{v u_{1}}$ to go back to the origin and additional time $\frac{d}{v}$ to go from the origin to the exit. Thus taking this into account the total evacuation time will be as follows:
\begin{align}
\label{eq:evac1}
    {\cal E}_1=\frac{d}{u_{1}}+\frac{d u_{2}}{v u_{1}}+\frac{d}{v}
\end{align}

\paragraph{Case~2:} $R$ finds the exit first.

Based on Fig. \ref{fig3}, 
\begin{figure}[h] 
\begin{center}
\epsfig{width=11cm,file=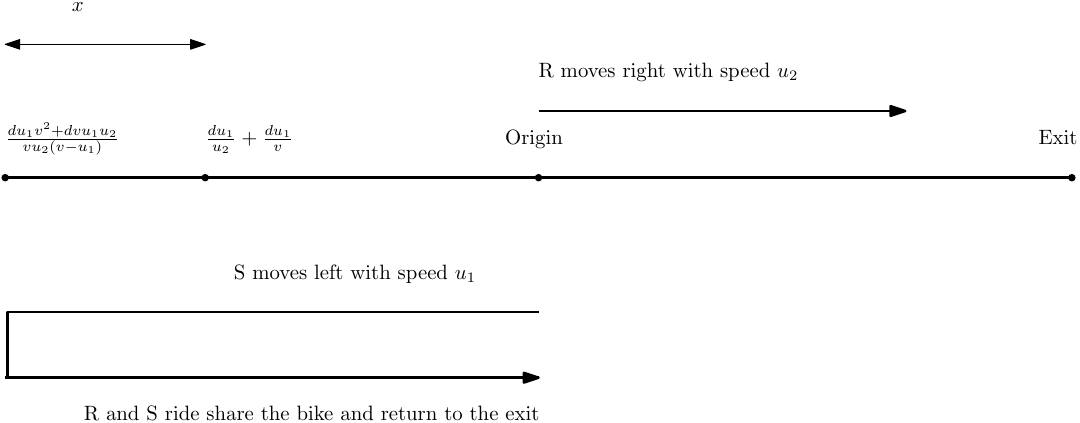}
\end{center}
     \caption{Depicted is Case 2 where two robots S and R are moving in opposite direction and such that R finds the exit.}
       \label{fig3}
\end{figure} 
when $R$ reaches the exit, then $S$ is at distance $\frac{d u_{1}}{u_{2}}$ to the left of the origin. In this case $R$ needs to go back to the left to bring $S$ toward the origin. When $R$ reaches the origin, $S$ will be at distance $\frac{d u_{1}}{u_{2}}+\frac{d u_{1}}{v}$ to the left of the origin. At that point if $S$ meets $R$ at distance $x$ from its current position, then we have the following:
\begin{align*}
    \frac{x}{u_{1}}=\frac{d u_{1}}{u_{2} v}+\frac{d  u_{1}}{v^2}+\frac{x}{v}
    &\implies x  \left( \frac{v-u_{1}}{u_{1} v} \right)=\frac{d u_{1} v+d u_{1} u_{2}}{u_{2} v^{2}}\\
    &\implies x=\frac{d u_{1}^{2} v+d u_{1}^{2} u_{2}}{(v-u_{1}) u_{2} v}\\
\end{align*}
After $S$ and $R$ meet, they will be away from the origin by the following distance:
\begin{align*}
\frac{d u_{1}}{u_{2}}+\frac{d u_{1}}{v}+x
&=\frac{d u_{1}}{u_{2}}+\frac{d u_{1}}{v}+\frac{d u_{1}^{2} v+d u_{1}^{2} u_{2}}{(v-u_{1}) u_{2} v}
=\frac{d u_{1} v^{2}+d v u_{1} u_{2}}{v u_{2}(v-u_{1})}\\
\end{align*}
Thus $S$ and $R$ will be away from the exit by $$d+\frac{d u_{1} v^{2}+d v u_{1} u_{2}}{v u_{2}(v-u_{1})}=\frac{d u_{1} v+d u_{2} v}{u_{2} (v-u_{1})}.$$
Using Lemma \ref{lm:bikeshare}, the time needed to share the bike at distance $\frac{d u_{1} v+d u_{2} v}{u_{2} (v-u_{1})}$ from the exit will be as follows:
\begin{align*}
    \frac{(d u_{1} v+d u_{2} v)(v+1)}{2 v u_{2}(v-u_{1})}=\frac{d u_{1} v^{2}+d u_{1} v+d u_{2} v^{2}+d u_{2} v}{2 u_{2} v (v-u_{1})}
\end{align*}
Now we can find the evacuation time ${\cal E}_2$ as follows:
\begin{align}
{\cal E}_2&=\notag
\frac{d}{u_{2}}+\frac{d}{v}+\frac{d u_{1}}{u_{2} v}+\frac{d u_{1}}{v^{2}}+\frac{d u_{1}^{2} v+d u_{1}^{2} u_{2}}{(v-u_{1}) u_{2} v^{2}}+\frac{d u_{1} v^{2}+d u_{1} v+d u_{2} v^{2}+d u_{2} v}{2 u_{2} v (v-u_{1})}\\
    &=\label{eq:evac2} 
    \frac{2 d v^{2}+2 d u_{2} v+d u_{1} v^{2}+d u_{1} v+d u_{2} v^{2}+d v u_{2}}{2 u_{2} v (v-u_{1})}.
\end{align}

Since we do not know where the exit is situated, we should consider the worst case scenario of the evacuation time, denoted by ${\cal E}$, to be the maximum of ${\cal E}_1$ and ${\cal E}_2$, as given in Equations~\eqref{eq:evac1}~and~\eqref{eq:evac2}, as follows:
\begin{align}
\label{eq:evactime1}
    {\cal E}&=\max\{{\cal E}_1,{\cal E}_2\} 
\end{align}
In order to minimize the total evacuation time, we need to know for what value of $v$ is each of ${\cal E}_1$ and ${\cal E}_2$ minimized.

Using Equation~\eqref{eq:evac1}, we see that ${\cal E}_1$ is minimized if $u_{1}=1$ and $u_{2}=0$ since we assume that the exit is on the side of the sender $S$. In this case ${\cal E}_1$ will be as follows: ${\cal E}_1=d\cdot \frac{1+v}{v}$.

Using Equation~\eqref{eq:evac2}, we see that ${\cal E}_2$ is minimized if $u_{2}=v$ and $u_{1}=0$ since we assume that the exit is on the side of the receiver $R$.
In this case ${\cal E}_2$ will be as follows:
$    {\cal E}_2=d\cdot \frac{v^{3}+5 v^{2}}{2 v^{3}}=d\cdot \frac{v+5}{2 v}$.
Thus 
$$
{\cal E}_1-{\cal E}_2=d\cdot (\frac{1+v}{v}-\frac{v+5}{2 v})=d\cdot\frac{v-3}{2 v}.
$$ 
We conclude that ${\cal E}_1\leq {\cal E}_2$, if $v\leq 3$. 
Thus we have two cases to consider based on the speed of the biker. 

\paragraph{Case~1:} $v\leq 3$
\newline We have ${\cal E}_1<{\cal E}_2$ and thus the maximum evacuation time is ${\cal E}=\max\{{\cal E}_1,{\cal E}_2\}={\cal E}_2$ and in this case ${\cal E}_2$ is minimized if we set $u_{2}=v$. 

\paragraph{Case~2:} $3\leq v\leq 10$ 
\newline (Note that if $v>10$, then the zig-zag algorithm discussed later in section \ref{subsec:alg3} performs better.) We have ${\cal E}_2<{\cal E}_1$ and thus the maximum evacuation time is ${\cal E}=\max\{{\cal E}_1,{\cal E}_2\}={\cal E}_1$ and in this case ${\cal E}_1$ is minimized when we set $u_{1}=1$. 

\subsubsection{Algorithm~\ref{algorithm1}: $1\leq v \leq 3$}
\label{subsec:alg1}
\paragraph{}In the first algorithm, the receiver takes the bike and moves in one direction with its maximum speed, while the sender moves with optimal speed $0\leq u_{1}\leq 1$ in the other direction. If the sender finds the exit first, then it informs the receiver and stays at the exit, while if the receiver finds the exit first, then it changes direction to catch up with the sender and then both robots ride-share the bike toward the exit. The first algorithm performs better when the speed of the bike satisfies $1\leq v \leq 3$. If we assume the sender moves to the left with speed $0\leq u_{1}\leq 1$ and the receiver to the right with maximum speed $v$. 

The algorithm will be as follows:
\begin{algorithm}[H]
\caption{OppDirectionWithBikerMovingAtMaxSpeed ($S$ source, $D$ destination)}\label{algorithm1}
\begin{algorithmic}[1]
\State {Sender moves to the left with speed $u_{1}$}
\State {Receiver takes the bike and moves to the right with speed $v$}
\If {Receiver finds the Exit first}
\State{it changes direction and catches up to the Sender;}
\State{The two robots change direction and ride-share the bike towards the Exit;}
\Else
\If {Sender finds the Exit first} 
\State{it informs the Receiver that the Exit has been found but stays at the Exit;}
\State{The Receiver rides the bike to the Exit;}
\EndIf
\EndIf
\end{algorithmic}
\end{algorithm}
We prove the following result.
\begin{theorem}
\label{theorem1}
The optimal competitive ratio of Algorithm~\ref{algorithm1} is $$\frac{2v}{v+1}  \left(\frac{2 v+\frac{-(7v+v^2)+v\sqrt{v^2+30 v+97}}{2 v+6}}{v \frac{-(7v+v^2)+v\sqrt{v^2+30 v+97}}{2 v+6}}\right)$$
\end{theorem}
\begin{proof}
We have that $S$ moves in one direction with speed $0\leq u_{1}\leq 1$ and $R$ moves in the other direction with speed $v$. Thus we will have the following equations
    ${\cal E}_1=\frac{2 d}{u_{1}}+\frac{d}{v}\:\:\:and\:\:\: {\cal E}_2=\frac{d v^{2}+5 d v+d u_{1} v+d u_{1}}{2 v (v-u_{1})}$.
In order to determine the optimal speed $u_{1}$ for the sender, we should set the evacuation times in both cases above to be equal, namely ${\cal E}_2 = {\cal E}_1$. This will give us the following:
\begin{align*}
    &\frac{d(v^2+5 v+u_{1}  v+u_{1})}{2 v (v-u_{1})}=\frac{2 d v+d u_{1}}{u_{1} v}
    \implies (3+v) u_{1}^{2}+(7 v+v^{2}) u_{1}-4 v^{2}=0
\end{align*}
This is a quadratic equation of $u_{1}$ and it has two solutions. Keeping the positive root we get 
$$
u_{1}=\frac{-(7v+v^2)+v\sqrt{v^2+30 v+97}}{2 v+6}.
$$ 
We expect the optimal speed $u_{1}$ to be less than $1$ when the speed $v$ is between $1$ and $3$.
If we substitute $u_{1}$ in the evacuation time, denoted by ${\cal E}$, calculated in any of the two cases we get the following
\begin{align*}
    {\cal E}=d \left(\frac{2 v+\frac{-(7v+v^2)+v\sqrt{v^2+30 v+97}}{2 v+6}}{v \frac{-(7v+v^2)+v\sqrt{v^2+30 v+97}}{2 v+6}}\right)
\end{align*}
Thus, using Lemma~\ref{lm:bikeshare}, the competitive ratio will be 
$$
\frac{2v}{v+1}  \left(\frac{2 v+\frac{-(7v+v^2)+v\sqrt{v^2+30 v+97}}{2 v+6}}{v \frac{-(7v+v^2)+v\sqrt{v^2+30 v+97}}{2 v+6}}\right).
$$
This completes the proof of Theorem~\ref{theorem1}. \qed
\end{proof}

\subsubsection{Algorithm~\ref{algorithm2}: $3\leq v\leq 10$}
\label{subsec:alg2}

\paragraph{} In the second algorithm, the  sender moves in one direction with unit speed while the receiver takes the bike and moves in the other direction with optimal speed $1<u_{2}<v$. If the sender finds the exit first, then it informs the receiver and stays at the exit, while if the receiver finds the exit first, then it changes direction to catch up with the sender and then both robots ride-share the bike toward the exit. The second algorithm performs better when the speed $3\leq v\leq 10$ as we will see later in figure \ref{Performance}. We assume the sender moves to the left with unit speed and the receiver to the right with speed $1\leq u_{2}\leq v$. 

The main algorithm will be as follows: 
\begin{algorithm}[H]
\caption{OppDirectionWithWalkerMovingAtMaxSpeed ($S$ source, $D$ destination)}\label{algorithm2}
\begin{algorithmic}[1]
\State {Sender moves to the left with unit speed}
\State {Receiver takes the bike and moves to the right with speed $u_{2}$}
\If {Receiver finds the Exit first}
\State{it changes direction and catches up to the Sender;}
\State{The two robots change direction and ride-share the bike towards the Exit;}
\Else
\If {Sender finds the Exit first} 
\State{it informs the Receiver that the Exit has been found but stays at the Exit;}
\State{The Receiver rides the bike to the Exit;}
\EndIf
\EndIf
\end{algorithmic}
\end{algorithm}
We prove the following result.
\begin{theorem}
\label{theorem2}
The optimal competitive ratio of Algorithm~\ref{algorithm2} is:
\begin{align*}
\frac{2v}{v+1} \left(1+\frac{1}{v}+\frac{v^{2}-3 v-2-\sqrt{v^{4}+18 v^{3}-7 v^{2}+4 v+4}}{4 v (1-v)}\right)
\end{align*}
\end{theorem}
\begin{proof}
 $S$ moves in one direction with unit speed and $R$ moves in the other direction with speed $3\leq u_{2}\leq v$. We will have the following:
\begin{align*}
    &{\cal E}_1=\frac{d v+d u_{2}+d}{v}
    \mbox{ and }
    {\cal E}_2=\frac{3 d v+d u_{2} v+3 d u_{2}+d}{2 u_{2} (v-1)}
\end{align*}
The optimal speed $u_{2}$ is obtained by setting the evacuation time in the two cases above to be equal, namely ${\cal E}_2= {\cal E}_1$. Thus we have the following:
\begin{align*}
  &\frac{u_{2}+v+1}{v}=\frac{3 v^2+3 u_{2} v+u_{2} v^{2}+v}{2 u_{2} v (v-1)}\\
\end{align*}
If we multiply out we will end up having the following quadratic equation in the variable $u_2$:
$
  (2-2 v) u_{2}^{2}+(3 v-v^{2}+2) u_{2}+3 v^{2}+v=0 .
$  
Keeping the positive root gives:
\begin{align*}
u_{2}=\frac{v^{2}-3 v-2-\sqrt{v^{4}+18 v^{3}-7 v^{2}+4 v+4}}{4 (1-v)}
\end{align*}
Thus the evacuation time will be 
\begin{align*}
d \left(1+\frac{1}{v}+\frac{v^{2}-3 v-2-\sqrt{v^{4}+18 v^{3}-7 v^{2}+4 v+4}}{4 v (1-v)}\right)
\end{align*}
Using Lemma~\ref{lm:bikeshare}, it follows that the competitive ratio will be 
\begin{align*}
\frac{2v}{v+1} \left(1+\frac{1}{v}+\frac{v^{2}-3 v-2-\sqrt{v^{4}+18 v^{3}-7 v^{2}+4 v+4}}{4 v (1-v)}\right)
\end{align*}
This completes the proof of Theorem~\ref{theorem2}. \qed
\end{proof}

\subsection{Algorithm~\ref{algorithm3}: $10< v$}
\label{subsec:alg3}

In the next Algorithm~\ref{algorithm3}, the sender takes the bike and uses a doubling strategy to search for the exit and moves a distance $2^{k}$ during the $k$-th iteration. The receiver also uses a doubling strategy but since it is moving with unit speed it will try to stay as close as possible to the biker. This can be achieved by having the receiver moves a distance $\frac{2^{k}}{v}$ during the $k^{th}$ iteration, since moving any further will cause the sender to be farther away from the receiver during the $(k+1)^{st}$ iteration. 

In the third algorithm the Receiver imitates the Sender and as it will be seen in the sequel it performs better when the speed $v$ satisfies $v>10$.  The main algorithm is as follows:
\newline
\begin{algorithm}[H]
 \caption{(ReceiverImitateSender)}
 \label{algorithm3}
\begin{algorithmic}[1]
 \For{$k \gets 1$ to $\infty$} 
  \If{k is odd (resp.even)}
  \State{Sender takes the bike and moves right (resp. left) a distance {$2^k$} unless the exit is found;}
  \State{Receiver moves right (resp. left) a distance {$\frac{2^k}{v}$};}
 \If{exit is found by sender}
     \State{Communicate with receiver};
     \State{Sender moves back $\frac{d}{2}-\frac{d}{2v}$ to leave the bike for the receiver and then returns to exit;}
     \State{Receiver continues toward the exit after picking up the bike;}
     \State{Quit;}
 \EndIf
   \State{Sender turns; then moves left (resp. right), returns to the origin;} 
   \State{Receiver turns; then moves left (resp. right), returns to the origin;}
   \EndIf
\EndFor
 \end{algorithmic}
\end{algorithm}
\begin{theorem}
\label{theorem3-ub}
The competitive ratio for Algorithm~\ref{algorithm3} is 
$
\leq \left(\frac{2 v}{v+1}\right) \left(\frac{9}{v}+\frac{1}{2}-\frac{1}{2v^2}\right).
$
\end{theorem}
\subsection{Performance of the Upper bound algorithms}
\label{subsec:alg123}

The performance of the three algorithms is compared in Fig.~\ref{Performance}. As it can be seen in this figure, the performance of the algorithm depends on the speed $v$ of the non-autonomous robot. 
\begin{figure}[h]
\begin{center}
\includegraphics[width=10cm,height=7cm]{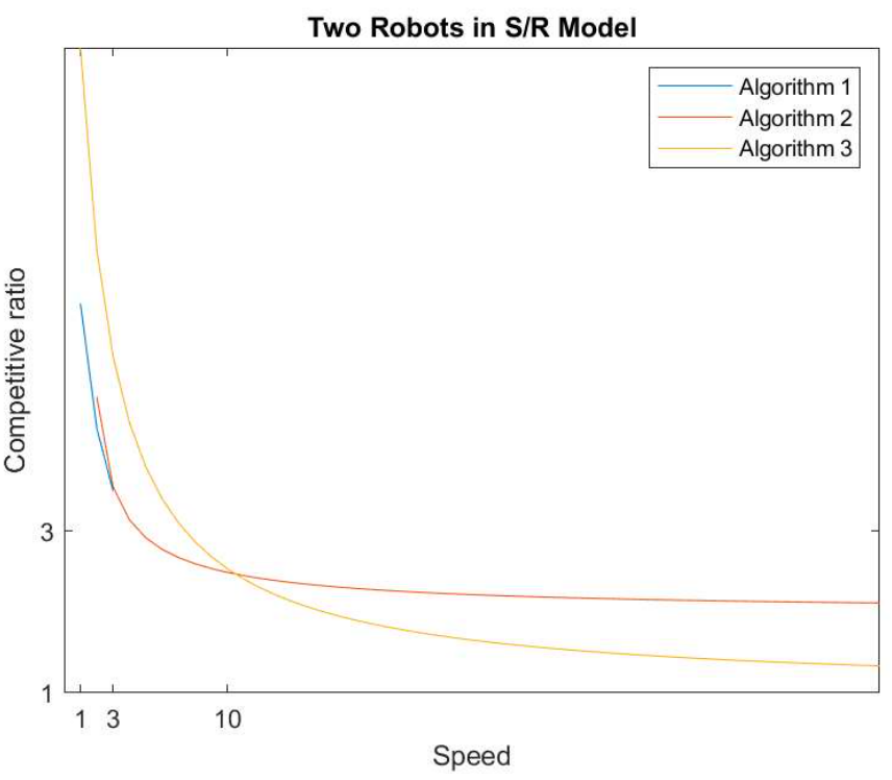}
\end{center}
\caption[Graph showing the performance of the 3 algorithms by showing how the competitive ratio fluctuates in terms of speed]{Graph showing the performance of the 3 algorithms by showing how the competitive ratio fluctuates in terms of the speed. Note in the pictures above Algorithms 1 is depicted in the range $[1,3]$ and Algorithms 2 and 3 are depicted in the range $[1,+\infty)$. }
  \label{Performance}
\end{figure} 

From the plots we can confirm that Algorithm~\ref{algorithm1} performs best if the speed is less than 3, while Algorithm~\ref{algorithm2} performs best if the speed is between $3$ and $10$. Finally, Algorithm~\ref{algorithm3} performs best if the speed is larger than 10 and in this case the competitive ratio converges to $1$ as $v \to \infty$, which is also the best that could be achieved if both autonomous robots have full knowledge in that they know where the exit is.

\section{Lower Bound}
\label{sec:Lower Bound}

Finally, we prove a lower bound.
\begin{theorem}
\label{theorem3-lb}
The lower bound for the competitive ratio is 
$$
\left\{
\begin{array}{cl}
\frac{6}{v+1} & \mbox{ if $v\leq 3$}\\
\frac{v^{2}+2 v-3}{v^{2}-1}  & \mbox{ if $v> 3$}
\end{array}
\right.
$$
\end{theorem}
 \section{Conclusion}
\label{sec:Conclusion}

The purpose of our investigations was to study the limits of communication for robots with mixed faults. In particular, in this paper we gave evacuation algorithms and studied the upper and lower bounds for evacuation in the S/R communication model where the faults involve a sender and a receiver seeking to evacuate through an unknown exit situated on an infinite line. 
In particular, in Section~\ref{sec:Lower Bound} we presented a lower bound and in subsection~\ref{subsec:alg123} we presented the asymptotic behaviour of the main algorithm. However the possibility of improving the lower bound remains an open problem. 
An extension of the problem already analyzed could also be considered for the case of multiple (more than two) autonomous robots in a setting where the autonomous robots may suffer any of S/R, crash, and byzantine faults. One could also consider the case of multiple non-autonomous robots. The possibility of studying other search domains (geometric star graph, disk, polygon, etc) would also be interesting.

\bibliographystyle{abbrv}
\bibliography{refs}

\begin{thebibliography}{10}

\bibitem{ahlswede1987search}
R.~Ahlswede and I.~Wegener.
\newblock {\em Search problems}.
\newblock Wiley-Interscience, 1987.

\bibitem{baezayates1993searching}
R.~Baeza~Yates, J.~Culberson, and G.~Rawlins.
\newblock Searching in the plane.
\newblock {\em Information and Computation}, 106(2):234--252, 1993.

\bibitem{BS95}
R.~Baeza-Yates and R.~Schott.
\newblock Parallel searching in the plane.
\newblock {\em Computational Geometry}, 5(3):143--154, 1995.

\bibitem{bampas2019linear}
E.~Bampas, J.~Czyzowicz, L.~Gasieniec, D.~Ilcinkas, R.~Klasing, T.~Kociumaka,
  and D.~Pajak.
\newblock Linear search by a pair of distinct-speed robots.
\newblock {\em Algorithmica}, 81(1):317--342, 2019.

\bibitem{beck1964linear}
A.~Beck.
\newblock On the linear search problem.
\newblock {\em Israel Journal of Mathematics}, 2(4):221--228, 1964.

\bibitem{bellman1963optimal}
R.~Bellman.
\newblock An optimal search.
\newblock {\em SIAM Review}, 5(3):274--274, 1963.

\bibitem{Bose16}
P.~Bose and J.-L. De~Carufel.
\newblock A general framework for searching on a line.
\newblock {\em {Theoretical Computer Science}}, pages 703:1--17, 2017.

\bibitem{brandt2017collaboration}
S.~Brandt, F.~Laufenberg, Y.~Lv, D.~Stolz, and R.~Wattenhofer.
\newblock Collaboration without communication: Evacuating two robots from a
  disk.
\newblock In {\em International Conference on Algorithms and Complexity}, pages
  104--115. Springer, 2017.

\bibitem{chrobak2015group}
M.~Chrobak, L.~Gasieniec, G.~T., and R.~Martin.
\newblock Group search on the line.
\newblock In {\em SOFSEM}, pages 164--176. Springer, 2015.

\bibitem{czyzowicz2018evacuating}
J.~Czyzowicz, S.~Dobrev, K.~Georgiou, E.~Kranakis, and F.~MacQuarrie.
\newblock Evacuating two robots from multiple unknown exits in a circle.
\newblock {\em Theoretical Computer Science}, 709:20--30, 2018.

\bibitem{czyzowicz2017evacuation}
J.~Czyzowicz, K.~Georgiou, M.~Godon, E.~Kranakis, D.~Krizanc, W.~Rytter, and
  M.~W{\l}odarczyk.
\newblock Evacuation from a disc in the presence of a faulty robot.
\newblock In {\em International Colloquium on Structural Information and
  Communication Complexity}, pages 158--173. Springer, 2017.

\bibitem{czyzowicz2019groupkos}
J.~Czyzowicz, K.~Georgiou, and E.~Kranakis.
\newblock Group search and evacuation.
\newblock In {\em Distributed Computing by Mobile Entities}, pages 335--370.
  Springer, 2019.

\bibitem{isaacCzyzowiczGKKNOS16}
J.~Czyzowicz, K.~Georgiou, E.~Kranakis, D.~Krizanc, L.~Narayanan, J.~Opatrny,
  and S.~Shende.
\newblock Search on a line by byzantine robots.
\newblock In {\em {ISAAC}}, pages 27:1--27:12, 2016.

\bibitem{czyzowicz2021groupevac}
J.~Czyzowicz, R.~Killick, E.~Kranakis, D.~Krizanc, L.~Narayanan, J.~Opatrny,
  D.~Pankratov, and S.~Shende.
\newblock Group evacuation on a line by agents with different communication
  abilities.
\newblock {\em ISAAC 2021}, pages 57:1--57:24, 2021.

\bibitem{czyzowicz2021searchnew}
J.~Czyzowicz, R.~Killick, E.~Kranakis, and G.~Stachowiak.
\newblock Search and evacuation with a near majority of faulty agents.
\newblock In {\em SIAM Conference on Applied and Computational Discrete
  Algorithms (ACDA21)}, pages 217--227. SIAM, 2021.

\bibitem{PODC16}
J.~Czyzowicz, E.~Kranakis, D.~Krizanc, L.~Narayanan, and J.~Opatrny.
\newblock Search on a line with faulty robots.
\newblock {\em Distributed Computing}, 32(6):493--504, 2019.

\bibitem{demaine2006online}
E.~D. Demaine, S.~P. Fekete, and S.~Gal.
\newblock Online searching with turn cost.
\newblock {\em Theoretical Computer Science}, 361(2):342--355, 2006.

\bibitem{georgiou2022evacuation}
K.~Georgiou, N.~Giachoudis, and E.~Kranakis.
\newblock Evacuation from a disk for robots with asymmetric communication.
\newblock In {\em 33rd International Symposium on Algorithms and Computation
  (ISAAC 2022)}. Schloss Dagstuhl-Leibniz-Zentrum f{\"u}r Informatik, 2022.

\bibitem{georgiou2019optimal}
K.~Georgiou, E.~Kranakis, N.~Leonardos, A.~Pagourtzis, and I.~Papaioannou.
\newblock Optimal circle search despite the presence of faulty robots.
\newblock In {\em International Symposium on Algorithms and Experiments for
  Sensor Systems, Wireless Networks and Distributed Robotics}, pages 192--205.
  Springer, 2019.

\bibitem{jawharkranakis21}
K.~Jawhar and E.~Kranakis.
\newblock Bike assisted evacuation on a line.
\newblock In {\em SOFSEM (47th International Conference on Current Trends in
  Theory and Practice of Computer Science)}. Springer, LNCS, 2021.

\bibitem{stone1975theory}
L.~Stone.
\newblock {\em Theory of optimal search}.
\newblock Academic Press New York, 1975.

\end{thebibliography}

\newpage

\appendix

\section{Proof of Lemma~\ref{lm:bikeshare}}

\begin{proof} (Lemma~\ref{lm:bikeshare})
Based on Fig. \ref{fig1}, both robots are initially situated at the origin which is at distance $d$ from the exit. 
\begin{figure}[h]  
\begin{center}
\epsfig{width=11cm,file=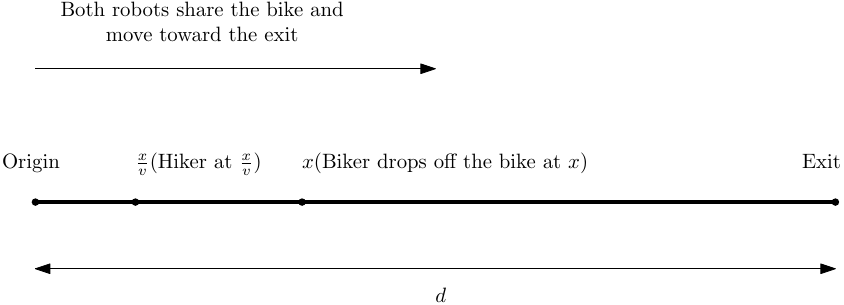}
\end{center}
\caption{Graph depicts the movement of two robots sharing the bike and starting at the origin in order to cover distance $d$. }
       \label{fig1}
\end{figure} 
The objective of both robots is to ride share the bike and arrive at the exit at the same time. If we assume the first robot drops off the bike at distance $x$ then we should have the following:
\begin{align*}
&\frac{x}{v}+d-x=x+\frac{d-x}{v}.
\end{align*}
Solving the last equation in $x$ we derive that $x=\frac d2$. Therefore the time it takes for the two robots to find the Exit is $\frac{x}{v}+d-x=\frac{d (v+1)}{2 v}$. For a proof of the lower bound see~\cite{jawharkranakis21}.
\qed
\end{proof}
 \section{Proof of Theorem ~\ref{theorem3-ub}}
\begin{proof} In this algorithm the sender uses a doubling strategy with maximum speed $v$. The receiver will follow the sender but will move $\frac{2^k}{v}$ in each iteration instead of $2^k$. The sender will reach the exit first then it will communicate with the receiver to proceed to the exit. The sender will go back distance $\frac{d}{2}-\frac{d}{2v}$ to drop off the bike so that the receiver can pick it up on its way to the exit. We will justify why the sender needs to move $\frac{d}{2}-\frac{d}{2v}$ after reaching the exit to leave the bike for the receiver. After the sender reaches the exit, there is no benefit to stay at the exit with the bike since the receiver which is moving with unit speed can benefit from the bike to reach the exit faster. 

The key to finding the distance $x$ which is the distance between the exit and the point where the bike is dropped off is to have the sender drop it off at a point such that when it goes back to the exit it will reach the exit at the same time as the receiver. 
\begin{figure}[H]
\begin{center}
\epsfig{width=10cm,file=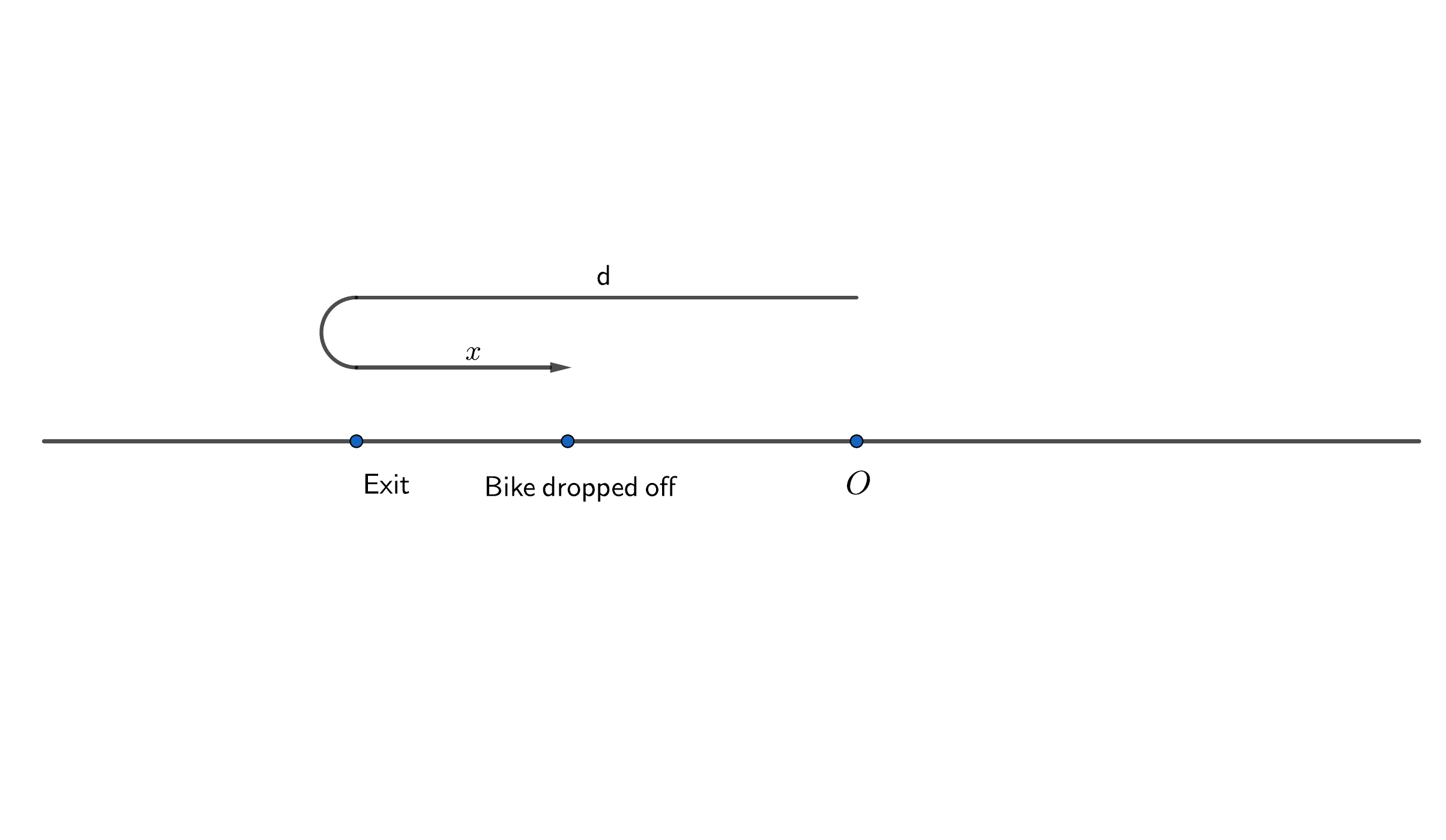}
\end{center}
  \caption{Graph depicts the minimal distance $x$ away from the origin at which the bike is dropped off by the sender and allows the receiver to take the bike and reach the exit sooner.}
  \label{diagramForAlgorithm3}
\end{figure} 
Based on Fig. \ref{diagramForAlgorithm3}, if we recall that $d$ is the distance from the origin to the exit and $x$ is the distance from the exit to the point where the sender drops off the bike, then we have the following:
\begin{align*}
    d-x+\frac{x}{v}=\frac{d}{v}+\frac{x}{v}+x
    &\implies x=\frac{d}{2}-\frac{d}{2v}\\
\end{align*} 
This will guarantee that when the sender drops off the bike at distance $x$, it will reach the exit at the same time as the receiver. 

Assume that the exit is found during the $k^{th}$ iteration, then  $2^{k-2} < d \leq2^k$. We can calculate the evacuation time, denoted by ${\cal E}$, as follows:
\begin{align*}
    {\cal E}&=\frac{2\cdot 2^0}{v}+\frac{2\cdot 2^1}{v}+\dots+\frac{2\cdot 2^{k-1}}{v}+d-x+\frac{x}{v}\\ 
    &=\frac{2(2^k-1)}{v}+d-x+\frac{x}{v}\\
    &=\frac{2^{k+1}}{v}-\frac{2}{v}+d-\frac{d}{2}+\frac{d}{2v}+\frac{d}{2v}-\frac{d}{2v^2} \\
    &=\frac{2^{k+1}}{v}-\frac{2}{v}+\frac{d}{2}+\frac{d}{v}-\frac{d}{2v^2}
\end{align*}

Combining similar terms above we obtain
\begin{align*}    
{\cal E} &\leq 2^3 \cdot \frac{2^{k-2}}{v}-\frac{2}{v}+\frac{d}{2}+\frac{d}{v}-\frac{d}{2v^2}\\
    &\leq \frac{8d}{v}-\frac{2}{v}+\frac{d}{2}+\frac{d}{v}-\frac{d}{2v^2}\\
    &\leq \frac{9d}{v}+\frac{d}{2}-\frac{d}{2v^2}-\frac{2}{v}\\
    &\leq \frac{9d}{v}+\frac{d}{2}-\frac{d}{2v^2}
\end{align*}
Thus, using Lemma~\ref{lm:bikeshare}, the resulting competitive ratio will be upper bounded by the quantity $\left(\frac{2 v}{v+1}\right) \left(\frac{9}{v}+\frac{1}{2}-\frac{1}{2v^2}\right).$
This completes the proof of Theorem~\ref{theorem3-ub}. \qed
\end{proof}
\section{Proof of Theorem~\ref{theorem3-lb}}

\begin{proof} (Theorem~\ref{theorem3-lb})
The Exit is placed at one of the endpoints $\pm d$ while the robots and the bike start at the origin $O$.  
\newline In order to find the lower bound we consider the worst case scenario which is achieved when the biker takes the bike and reaches $-d$ and the exit is at $d$. Depending on the speed of the biker, there are two cases to consider as follows:

\paragraph{Case 1:} $v\geq 3$

As we can see in Fig.~\ref{fig4}, if the biker is fast then it can reach the exit before the other robot and this occurs if the speed $v\geq 3$. 
\begin{figure}
\begin{center}
\epsfig{width=11cm,file=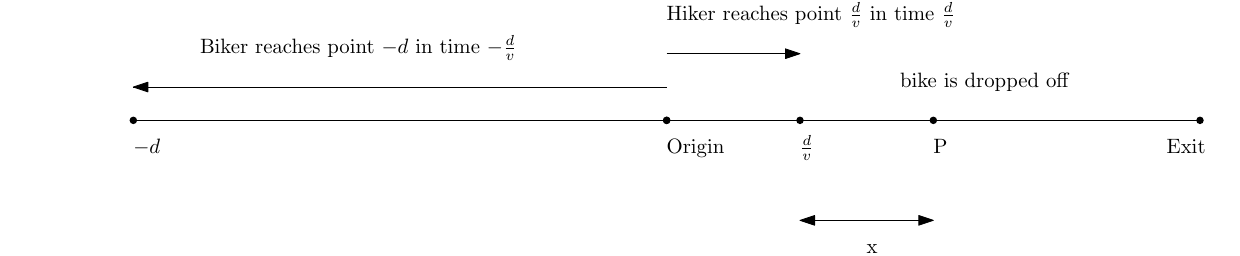}
\end{center}
\caption{ To achieve the optimal lower bound, both robots move in opposite direction. When the biker reaches $-d$ in time $\frac{d}{v}$, the adversary places the exit at $d$. At time $\frac{d}{v}$, the hiker would be at point $\frac{d}{v}$.}
  \label{fig4}
\end{figure} 
In this case the biker needs to share the bike with the other robot. The closest point $P$ to the exit that the other robot can reach is $\frac{d}{v}$ to the right of the origin. In order to find the evacuation time, we assume that the biker meets the other robot at distance $x$ to the right of point $P$. In order to achieve minimal evacuation time, we have the following:
\begin{align*}
        &\frac{d}{v}+\frac{d}{v^2}+\frac{x}{v}+d-x-\frac{d}{v}=x+\frac{1}{v} (d-\frac{d}{v}-x)\\
        &\implies \frac{d}{v}+\frac{d}{v^2}+\frac{x}{v}+d-\frac{d}{v}-x=x+\frac{d}{v}-\frac{d}{v^2}-\frac{x}{v}\\
        &\implies 2 x-\frac{2 x}{v}=\frac{2 d}{v^2}+d-\frac{d}{v}\\
       &\implies x=\frac{2 d+d v^{2}-d v}{2 v (v-1)}
\end{align*}

As a consequence, if ${\cal E}$ denotes the evacuation time then we have
\begin{align*}
{\cal E}&=\frac{d}{v}+x+\frac{d}{v}-\frac{x}{v}-\frac{d}{v^2}\\
       &=\frac{d}{v}+\frac{d (v^2-v+2)}{2 v (v-1)}+\frac{d}{v}-\frac{d}{v^{2}}-\frac{d v^2-d v+2 d}{2 v^{3}-2 v^{2}}\\
       &=\frac{d(v^{2}+2 v-3)}{2 v (v-1)}
\end{align*}

\paragraph{Case 2:} $v\leq 3$ 

In this case the biker cannot arrive to the exit before the other robot and therefore the minimal evacuation time will be $\frac{3d}{v}$   
  
Using Lemma~\ref{lm:bikeshare}, we conclude that a lower bound on the competitive ratio, denoted by $CR$, will satisfy the following inequalities:
 \paragraph{Case 1}: $v\geq 3$
 \begin{align}
 \label{eq:lb1}
 \frac{v^{2}+2 v-3}{v^{2}-1}\leq \frac{d (v^{2}+2 v-3)}{2 v (v-1)} \frac{2 v}{d (v+1)}\leq CR
 \end{align}
 
\paragraph{Case 2}: $v\leq 3$
\begin{align}
\label{eq:lb2}
\frac{6}{v+1} \leq  \frac{3 d}{v}  \frac{2 v}{d (v+1)}\leq CR 
 \end{align}
 
The lower bound follows from Inequalities~\eqref{eq:lb1}~and~\eqref{eq:lb2}. 
 \qed
 \end{proof}
\end{document}